\documentclass[a4paper]{article}

\usepackage{url}

\usepackage{amsmath}
\usepackage{amssymb}
\usepackage{amsthm}
\usepackage{latexsym}

\usepackage{tikz}
\usetikzlibrary{positioning}
\usetikzlibrary{shapes}
\usetikzlibrary{automata}
\usetikzlibrary{backgrounds}
\usetikzlibrary{fit}
\usetikzlibrary{decorations.pathmorphing}
\usetikzlibrary {arrows.meta}

\theoremstyle{plain}
\newtheorem{theorem}{Theorem}
\newtheorem{lemma}[theorem]{Lemma}
\newtheorem{corollary}[theorem]{Corollary}

\theoremstyle{definition}
\newtheorem{definition}[theorem]{Definition}
\newtheorem{example}[theorem]{Example}

\newcommand{\Prop}{\mathsf{Prop}}
\newcommand{\Ag}{\mathsf{Ag}}
\newcommand{\Ac}{\mathsf{Ac}}
\newcommand{\pos}{\mathsf{V}}
\newcommand{\E}{\mathsf{E}}
\newcommand{\val}{\ell}
\newcommand{\pow}{\mathcal{P}}
\newcommand{\I}{\mathrm{I}}
\newcommand{\K}{\mathsf{K}}
\newcommand{\CK}{\mathsf{C}_G}
\newcommand{\Str}{\mathsf{Str}}
\newcommand{\last}{\mathsf{last}}
\newcommand{\Hist}{\mathsf{Hist}}
\newcommand{\CGS}{\mathcal{G}}

\newcommand{\X}{\mathsf{X}}
\newcommand{\kd}{\mathsf{kd}}
\newcommand{\len}{\mathsf{len}}

\newcommand{\lang}{\mathcal{L}}
\newcommand{\langCK}{\lang^{\mathsf{C}}}

\title{Knowledge and Common Knowledge of Strategies\thanks{Research supported by the Swiss National Science Foundation project no.~227544 \emph{Epistemic Group Attitudes}}} 

\author{Borja Sierra Miranda \and Thomas Studer}
\date{\small{Institute of Computer Science\\University of Bern\\Switzerland}}

\begin{document}

\maketitle

\begin{abstract}
Most existing work on strategic reasoning simply adopts either an informed or an uninformed semantics. 
We propose a model where knowledge of strategies can be specified on a fine-grained level.  In particular,  it is possible to 
distinguish first-order, higher-order, and common knowledge of strategies.
We illustrate the effect of higher-order knowledge of strategies by studying the game Hanabi.  Further, we show that common knowledge of strategies is necessary to solve the consensus problem.  Finally, we study the decidability of the model checking problem.
\end{abstract}


\section{Introduction}

Strategic reasoning is a fundamental aspect of decision-making in multi-agent systems, game theory, and artificial intelligence. It involves the ability of agents to anticipate the actions of others, formulate plans, and achieve desired outcomes within competitive or cooperative settings. 
%
To formally analyze strategic interactions, researchers have developed logical frameworks such as 
Coalition Logic~\cite{10.1093/logcom/12.1.149},  Alternating-time Temporal Logic~\cite{10.1145/585265.585270},  and Strategy Logic~\cite{CHATTERJEE2010677,10.1145/2631917}.
Many of these systems have been extended to allow for epistemic reasoning~\cite{ijcai2021p246,DBLP:conf/kr/MaubertM18,HoekW03a}.   

Most works adopt one of two semantics regarding knowledge of strategies: in the uninformed semantics, agents do not know each other's strategies, e.g.~\cite{DBLP:conf/kr/MaubertM18,HoekW03a}, while in the informed one,  agents know everyone's strategies, e.g.~\cite{Saffidine_Schwarzentruber_Zanuttini_2018}.
The uninformed semantics models situations where agents have neither a priori knowledge of other agents' strategies nor any means to communicate information about their strategies.  In the informed semantics, on the other hand,  the agents' strategies are common knowledge,  which makes it possible for the agents to infer additional information from the observed actions.

A notable exception to the above-mentioned dichotomy is~\cite{ijcai2021p246}, which introduces a novel semantics for strategy logic where one can specify, for each agent $a$, the set of agents $A_a$ whose strategy $a$ is informed of.  
Hence, it captures the case of informed semantics,  the case of uninformed semantics, and intermediate cases where, e.g., some agents may privately communicate their strategies to other agents.

The logical presentation of strategic games introduced in~\cite{Artemov2014} also  includes knowledge of strategies via formulas of the form $\K_a s_j^l$ stating \emph{$a$ knows that $j$ has chosen her strategy~$l$}.  Artemov uses this presentation to study whether a game has a definitive solution under given epistemic/rationality conditions. Since the paper does not introduce a concrete language or deductive system,  questions of expressivity, decidability or complexity are not  discussed. 

\textbf{Contribution.}
We introduce a noval approach where knowledge of strategies can be specified on a  fine-grained level.
In particular, our formulation makes it possible to distinguish first-order and various forms of higher-order knowledge of strategies. Moreover, we can express that \emph{everybody knows} someone's strategy but also that a strategy is \emph{common knowledge} among a group of agents~\cite{fhmv95}.

Our semantics is based on concurrent game structures.  In order to specify (higher-order) knowledge of an agent about other agents' strategies, we introduce the notion of an \emph{information perspective} that can be seen as a generalization of the sets $A_a$ from ~\cite{ijcai2021p246} mentioned before. An information perspective is a set $\I$ containing sequences of agents.  
Such sequences are interpreted as follows: if, e.g.,  $abc \in \I$,  then agent $a$ knows that agent $b$ knows agent $c$'s strategy.
The logical language we evaluate over these structures is the usual language of epistemic logic extended with a temporal \emph{next}-operator.  This is sufficient for our examples and to investigate the power and the effect of information perspectives. 

We present two small examples showing the difference between having no knowledge and first-order knowledge of an agent's strategy.  These examples will be useful to investigate basic epistemic properties such as positive and negative introspection and reflection.
Moreover,  these initial examples also show that the concurrent game structure is general background knowledge for all agents.  The same phenomenon occurs in traditional epistemic logic where the Kripke structure is common knowledge among all agents~\cite{Artemov2024}.

Hanabi is a cooperative game with imperfect information.  We consider epistemic situations in this game to investigate the role of higher-order knowledge of strategies.  First, we show that if player $b$ knows player $a$'s strategy,  then $b$ can obtain new information from observing $a$'s actions. 
Further,  if $a$ knows that $b$ knows $a$'s strategy,  then $a$ knows that $b$ will obtain new information from observing $a$'s action.  In the situation that we study, this knowledge is a prerequisite for $a$ to safely perform the given action.  This illustrates how higher-order knowledge is necessary for successful gameplay.

Since an information perspective may be infinite, it is possible to express common knowledge of strategies.  We show that if all strategies are common knowledge,  then the individual knowledge operators satisfy the negative introspection property.
Then, we turn to the binary consensus problem~\cite{CachinBook}. 
We show that common knowledge of strategies is required to solve this task successfully.

Finally, we investigate the problem of model checking where we restrict ourselves to formulas that do not contain common knowledge operators.
The key technical lemma shows that, in this case, only information perspectives of bounded size have to be considered.  This then yields the decidability of the corresponding model checking problem.

\section{The model}\label{sec:2}

We start with a countable set of atomic propositions  $\Prop$ and a finite set of agents~$\Ag$.   We use $\pow(X)$ to denote the power set of a set $X$ and $X^Y$ to denote the set of functions from $Y$ to $X$.

\begin{definition}
A \emph{concurrent game structure} is a tuple
$
(\Ac,  \pos, \E, \val, \sim_a)
$
where 
\begin{enumerate}
\item $\Ac$ is a finite set of actions,
\item $\pos$ is a finite set of positions,
\item $\E: \pos \times \Ac^{\Ag} \to \pos$ is a transition function,
\item $\val: \pos \to \pow(\Prop)$ a valuation function,
\item $\sim_a \,\subseteq\, (\pos \times \pos) \cup (\Ac \times \Ac)$ is an equivalence relation for each agent  $a\in \Ag$,  called $a$'s observation relation.
\end{enumerate}
\end{definition}

\begin{definition}
A \emph{joint action} is a function $\alpha: \Ag \to \Ac$ mapping each agent to an action. 
\end{definition}

\begin{definition}
A \emph{history} is a finite sequence of positions and joint actions 
\[
\rho = v_0 \alpha_1 v_1 \ldots \alpha_n v_n
\]
such that $\E(v_i, \alpha_{i+1})=v_{i+1}$. 
For a history $\rho = v_0 \alpha_1 v_1 \ldots \alpha_n v_n$ and $i\leq n$,  we set $\rho_{\leq i} := v_0 \alpha_1 v_1 \ldots \alpha_i v_i$ and $\last(\rho):=v_n$.
We let $\Hist$ be the set of all histories.
\end{definition}

\begin{definition}
Two positions $v,v' \in \pos$ are \emph{indistiguishable} for an agent $a$ if $v \sim_a v'$,  and similarly for actions. 
Two joint actions $\alpha$ and $\beta$ are indistinguishable for agent $a$,  in symbols  $\alpha \sim_a \beta$ if they are pointwise indistinguishable for $a$,  i.e.~for all $b \in \Ag$, we have that $\alpha(b) \sim_a \beta(b)$.
\end{definition}

We assume synchronous perfect recall.  That is, agents remember their observations of all past positions and actions.  Hence, we extend the observation relations to histories as follows. 
\begin{definition}
Two histories $\rho = v_0 \alpha_1 v_1 \ldots \alpha_n v_n$ and $\rho' = v'_0 \alpha'_1 v'_1 \ldots \alpha'_m v_m$ are \emph{indistinguishable} for agent $a$,  in symbols $\rho \sim_a \rho'$,  if $m=n$,   $v_i \sim_a v'_i$ for all $0 \leq i \leq n$,  and  $\alpha_i \sim_a \alpha'_i$ for all $1 \leq i \leq n$.
\end{definition}

\begin{definition}
A \emph{strategy} is a function $\sigma: \Hist \to \Ac$.
We let $\Str$ denote the set of all strategies.
\end{definition}

\begin{definition}
An \emph{assignment} $\chi: \Ag \to \Str$ is a mapping of agents to strategies.
A history $\rho = v_0 \alpha_1 v_1 \ldots \alpha_n v_n$ is \emph{consistent} with an assignment $\chi$ for a set of agents $X$ if
\[
\alpha_{i+1}(b) = \chi(b)(\rho_{\leq i}) \quad \text{for all $i<n$ and agents $b \in X$.}
\]
\end{definition}

\begin{definition}
Given a concurrent game structure $\CGS$,  an assignment $\chi$, and a history $\rho = v_0 \alpha_1 v_1 \ldots \alpha_n v_n$, we define 
the \emph{one-step continuation of $\rho$} 
\[
\X_\CGS^\chi\rho := v_0 \alpha_1 v_1 \ldots \alpha_n v_n \alpha_{n+1} v_{n+1}
\]
 where  the joint action $\alpha_{n+1}$ is given by $\alpha_{n+1}(b):=\chi(b)(\rho)$ for all agents $b$ and $v_{n+1}:=\E(v_n,\alpha_{n+1})$ is the next position.
\end{definition}

Let $X$ be a finite set.
We use $X^{*}$ to denote the set of all finite words with symbols from~$X$
without adjacent repetitions,  that is if $x_n\cdots x_1 \in X^{*}$, then $x_i \neq x_{i+1}$ for all $1 \leq i < n$.

For a word $w = x_n\cdots x_1 \in X^{*}$ and an element \mbox{$y \in X$}, we write $yw$ for the word $y x_n\cdots x_1$.
The  \emph{length of a word~$w$},  $\len(w)$, is the number of symbols of~$w$.
That is,  for a word $w = x_n\cdots x_1 \in X^{*}$, we have $\len(w)=n$.
Further,  we set  
\[
X^{\geq n} := \{ w \in X^{*} \ |\ \len(w) \geq n\},
\]
that is $X^{\geq n}$ is the set of those words of   $X^{*}$ with a length of at least~$n$. 

An information perspective stores which agents are informed about the strategies of which other agents.
\begin{definition}
An \emph{information perspective} $\I$ is  a subset of $\Ag^{\geq 2}$.
Further, for an agent~$a$,  we set 
\[
\I_a :=\{ b \in \Ag \ |\ ab \in \I \} \cup \{ a \}
\]
and
\[
\I[a] := \{w \in \Ag^{\geq 2} \ |\ aw \in \I\} \ \cup
	\{ w \in \I \ |\ w= aw'  \text{ for some $w' \in \Ag^*$}\} .
\]
\end{definition}
Let $a,b,c \in \Ag$ be agents.   We interpret $ab \in \I$ as \emph{agent~$a$ knows (is informed about) agent $b$'s strategy}.  Hence, $\I_a$ is the set of all agents whose strategies agent $a$ is informed about.  
We assume that each agent knows its own strategy.   Therefore, we let $a \in \I_a$ by definition.

Further,  $abc \in \I$ means that agent $a$ knows that agent~$b$ knows agent $c$'s strategy,  and so on.
The operation $\I[a]$ changes the information perspective to agent $a$'s point of view.
That is, $\I[a]$ contains the information that $a$ knows.
We assume positive introspection, that is,  agents know about their own knowledge.
This is reflected by the condition that for all $w$ with $w= aw'$ for some $w'$,  we have
$w \in \I$ implies $w \in \I[a]$.


\begin{definition}
An information perspective $\I$ is \emph{truthful} if it is closed under suffixes,
i.e.~for any $a \in \Ag$ and any $w \in \Ag^{\geq 2}$ we have that 
\begin{equation}\label{eq:clsuf:1}
aw \in \I \text{ implies } w \in \I.
\end{equation}
\end{definition}
This means,  for instance,   if agent $a$ knows that agent~$b$ knows agent~$c$'s strategy,  then agent~$b$ must know agent~$c$'s strategy. 
Formally,    \eqref{eq:clsuf:1} guarantees that $abc \in \I$ implies \mbox{$bc \in \I$}.

On a technical level,  an  important consequence of \eqref{eq:clsuf:1} is that for any  truthful  information perspective $\I$ and any agent~$a$, we have
\begin{equation}\label{eq:truthful:1}
\I[a] \subseteq \I.
\end{equation}
This will be essential when we establish a version of the knowledge axiom T, see Lemma~\ref{l:refl:1}.

An agent $a$ cannot distinguish two assignments $\chi$ and $\chi'$ given an information perspective~$\I$ if $\chi$ and $\chi'$ agree on the strategies of all agents that agent~$a$ is informed about according to $\I$.
As mentioned before,  this includes the agent's own strategy.
\begin{definition}\label{d:chi_eq:1}
We say that two assignments $\chi$ and $\chi'$ are \emph{indistiguishable} for agent $a$ under the information perspective $\I$,  in symbols $\chi \sim^\I_a \chi'$,  if we have $\chi(b) = \chi'(b)$ for all $b \in \I_a$.
\end{definition}

\begin{definition}
A \emph{state} in a concurrent game structure is a triple $(\chi,\I,\rho)$ where
\begin{enumerate}
\item $\chi$ is an assignment, 
\item $\I$ is an information perspective,  and
\item $\rho$ is a history.
\end{enumerate}
Let $a$ be an agent. We say that $(\chi,\I,\rho)$ is \emph{$a$-consistent} if $\rho$ is consistent with~$\chi$ under $\I_a$.

Further, we call a state $(\chi',\I', \rho')$  \emph{accessible from  $(\chi,\I, \rho)$ for $a$},  in symbols 
$
(\chi,\I, \rho)  \trianglelefteq_a (\chi',\I',\rho'),
$
 if
\begin{enumerate}
\item $\chi \sim^\I_a \chi'$,
\item $\I[a] \subseteq \I'$,  
\item $\rho \sim_a \rho'$,
\item $(\chi,\I,\rho)$ is $a$-consistent  and
\item $(\chi',\I',\rho')$ is $a$-consistent.
\end{enumerate}
\end{definition}

In order to deal with common knowledge of a group of agents,  we need the following definitions where $G\subseteq \Ag$:
\begin{align*}
Z  \trianglelefteq_G Z' &\text{ if{f} } Z  \trianglelefteq_a Z'  \text{ for some $a\in G$},\\
Z  \trianglelefteq^{*}_G Z' &\text{ if{f} } Z  \trianglelefteq_G Z_1   \trianglelefteq_G \cdots  \trianglelefteq_G Z_n \trianglelefteq_G Z' \quad\text{for some states $Z_1,\ldots,Z_n$}.
\end{align*}
Hence $Z  \trianglelefteq^{*}_G Z'$ states that $Z'$ is reachable from $Z$ using agents from $G$.  Its definition includes the case that $Z  \trianglelefteq^{*}_G Z'$ if  $Z  \trianglelefteq_a Z'$ for some $a\in G$.

In the following, we study some basic properties of the accessibility relation.
First, we observe that replacing the history in a state $Z$ with a history that cannot be distinguished by agent $a$ gives a state $Z'$ that is accessible from $Z$ for $a$ (given that $Z$ and $Z'$ are $a$-consistent and the information perspective of $Z$ is truthful).

\begin{lemma}\label{l:consistency:2}
Let $\CGS=(\Ac,  \pos, \E, \val, \sim_a)$ be a concurrent game structure and let $a$ be an agent. 
We let $(\chi,\I,\rho)$ be an $a$-consistent state with a truthful information perspective $\I$.
For each history $\rho'$ such that $\rho \sim_a \rho'$ and $\rho'$ is consistent with $\chi$ under $\I_a$,  we have
\[
(\chi,\I,\rho)   \trianglelefteq_a  (\chi,\I,\rho').  
\]
\end{lemma}
\begin{proof}
This follows immediately from $\chi \sim^\I_a \chi$ and \eqref{eq:truthful:1}.
\end{proof}

The next lemma states that if $Z  \trianglelefteq_a Z'$,  then the history of the state 
$Z'$ is consistent from the point of view of agent~$a$ in $Z$.   This is related to the notation $\rho \sim^{\chi,\I}_a \rho'$ used  in~\cite{ijcai2021p246}.
\begin{lemma}\label{l:consistency:1}
Let $\CGS=(\Ac,  \pos, \E, \val, \sim_a)$ be a concurrent game structure and let
$Z=(\chi,\I,\rho)$ and $Z'=(\chi',\I',\rho')$ be states with $Z  \trianglelefteq_a Z'$ for some agent $a$.
Then $\rho'$ is consistent with $\chi$ under $\I_a$.
\end{lemma}
\begin{proof}
We have 
\[
\alpha'_{i+1}(b) = \chi'(b)(\rho'_{\leq i}) \quad \text{for all $i<n$ and agents $b \in  \I'_a$.}
\]
By $\I[a] \subseteq \I'$ we get
\[
\alpha'_{i+1}(b) = \chi'(b)(\rho'_{\leq i}) \quad \text{for all $i<n$ and agents $b \in  \I_a$.}
\]
Because of $\chi \sim^\I_a \chi'$,  we conclude 
\[
\alpha'_{i+1}(b) = \chi(b)(\rho'_{\leq i}) \quad \text{for all $i<n$ and agents $b \in  \I_a$,}
\]
which means that $\rho'$ is consistent with $\chi$ under $\I_a$.
\end{proof}


We now introduce our languages $\lang$  and $\langCK$,  which  includes formulas:
\begin{enumerate}
\item
$\K_a \phi$ meaning \emph{agent $a$ knows that $\phi$};
\item 
$\CK \phi$ meaning \emph{it is common knowledge among $G$ that $\phi$};
\item 
$\X \phi$ meaning \emph{at the next point in time $\phi$ holds}.
\end{enumerate}

\begin{definition}
Formulas of the language $\lang$ are given by the following grammar
\[
\phi ::= p \ |\ \bot \ |\ \phi \to \phi \ |\ \K_a \phi \ |\  \X \phi ,
\]
and formulas of the language $\langCK$ are given by 
\[
\phi ::= p \ |\ \bot \ |\ \phi \to \phi \ |\ \K_a \phi \ |\ \CK \phi \ |\  \X \phi ,
\]
where $p \in \Prop$,  $a \in \Ag$,  and $G \subseteq \Ag$.
The remaining Boolean connectives are defined as usual.
\end{definition}


The next-operator is the only temporal operator that we include in our language, as it is the only one needed for our examples.
It is straightforward to include further temporal operators, see, e.g.,~\cite{ijcai2021p246} for the case of the until-operator.

We define truth of a formula at a state in a concurrent game structure as follows.
\begin{definition}\label{def:truth:1}
Let $\CGS=(\Ac,  \pos, \E, \val, \sim_a)$ be a concurrent game structure and 
$Z=(\chi,\I,\rho)$ be a state.
\emph{Truth} of a formula $\phi$ in $\CGS$  at state $Z$,  in  symbols $\CGS,Z \Vdash \phi$, is inductively defined as follows:
\begin{enumerate}
\item $\CGS, Z \Vdash p$ if{f} $p \in \ell(\last(\rho))$  where $p \in \Prop$,
\item $\CGS, Z \nVdash \bot$,
\item $\CGS, Z \Vdash \phi \to \psi$  if{f}  $\CGS, Z \nVdash \phi$ or  $\CGS, Z \Vdash  \psi$, 
\item $\CGS, Z \Vdash \K_a \phi $  if{f}  $\CGS, Z' \Vdash \phi$  for all states $Z'$ with $Z \trianglelefteq_a Z'$,
\item $\CGS, Z \Vdash \CK \phi $  if{f}  $\CGS, Z' \Vdash \phi$  for all states $Z'$ with $Z \trianglelefteq^{*}_G Z'$,
\item $\CGS, Z \Vdash \X \phi $  if{f}  $\CGS, Z' \Vdash \phi$ where 
				$Z'=(\chi,\I,\X_\CGS^\chi\rho)$.
\end{enumerate}

We say that a formula $\phi$ is \emph{valid} if $\CGS,Z \Vdash \phi$ holds for all concurrent game structures $\CGS$ and states $Z$.
\end{definition}


\section{First examples}\label{s:firstExamples}

We give two basic examples to illustrate our truth definition.  These examples will also be referred to later in the proof of Lemma~\ref{l:negIntro:1}.

We consider a set $\Ag$ of two agents $a$ and $b$  and 
the concurrent game structure shown in Figure~\ref{fig:CGS1}.
In position $v_1$, only agent $a$ will perform an action. 
This can can be either action $\alpha$ leading to position $v_2$ where $p$ is true or 
action $\beta$ leading to position $v_3$ where $q$ is true.
We assume that every agent can distinguish the different positions and also the two actions.
Since the other agents' actions do not matter in these examples,  we will not mention them and consider $\alpha$ and $\beta$ also as joint actions. 
We also omit the loops on final positions.
\begin{figure}[ht]
\begin{center}
\begin{tikzpicture}[scale=0.45]
\node[circle, draw] at (0,0) (v1) [label = below:$ $]{$v_1$};
\node[circle, draw] at (-3,-4) (v2) [label = below:$p$]{$v_2$};
\node[circle, draw] at (3,-4) (v3) [label = below:$q$]{$v_3$};

\draw[-] (v1) to (v2);
\draw[-] (v1) to (v3);

\node[circle, fill = white] at (-2.5, -1.5) (l1)[label = above:$ $]{$\alpha$};
\node[circle, fill = white] at (2.5, -1.5) (l1)[label = below:$ $]{$\beta$};

\end{tikzpicture}
\caption{The concurrent game structure $\CGS_1$}\label{fig:CGS1}
\end{center}
\end{figure}
Let the strategy $\sigma_\alpha$ be such that $\sigma_\alpha(v_1) = \alpha$ and let~$\chi_\alpha$ be an assignment with $\chi_\alpha(a)=\sigma_\alpha$.
Similarly,  let  the strategy $\sigma_\beta$ be such that $\sigma_\beta(v_1) = \beta$ and let~$\chi_\beta$ be an assignment with $\chi(a)=\sigma_\beta$.

%
%


For the first example,  we assume that agent $b$ is informed about agent $a$'s strategy. 
\begin{example}\label{ex:basic:1}
We set $\I := \{ ba \}$ and consider the state $Z=(\chi_\alpha,\I,v_1)$.
In this setting,  we find that $\CGS_1, Z \Vdash \K_b \X p$ holds.

Indeed,  let $Z'=(\chi',\I',v_1)$ be an arbitrary state with $Z \trianglelefteq_b Z'$.
This implies $\chi_\alpha \sim^{\I}_b \chi'$, which by  $a \in \I_b$  yields  $\chi_\alpha(a) = \chi'(a)$.
Therefore,  we get  $\X_{\CGS_1}^{\chi'} v_1 = v_1 \alpha v_2$ and thus
$\CGS_1, Z' \Vdash \X p$.
Since~$Z'$ was arbitrary with $Z \trianglelefteq_b Z'$,  we conclude 
$\CGS_1, Z \Vdash \K_b \X p$.
\end{example}
For the second example, we consider the case when $b$ is not informed about $a$'s strategy. 
%
\begin{example}\label{ex:basic:2}
We set $\I=\emptyset$,  that is, in particular, $ba \notin \I$,
and consider the state $Z=(\chi_\alpha,\I,v_1)$.
We show that 
\begin{equation}\label{eq:ex:basic:2}
\CGS_1, Z \nVdash \K_b \X p.
\end{equation}
Indeed,  we find $\chi_\alpha \sim^\I_b \chi_\beta$.
Hence $Z \trianglelefteq_b Z'$ for $Z'=(\chi_\beta,\I,v_1)$.
We get $\CGS_1, Z' \Vdash \X \lnot p$ and thus conclude~\eqref{eq:ex:basic:2}.
\end{example}

Similarly, we can also establish $\CGS_1, Z \nVdash \K_b \X q$.
However, let us mention that, as usual, the concurrent game structure is general background knowledge.  That means agent $b$ knows that agent $a$ will perform either action $\alpha$ or action $\beta$ and thus reach a position where $p$ holds or one where $q$ holds.  Formally, we have $\CGS_1, Z \Vdash \K_b (\X p \lor \X q)$.

\section{Properties of knowledge}

In this section,  we examine basic properties of the individual knowledge operators.
We begin with positive introspection.

\begin{lemma}\label{l:ax4:1}
The formula $\K_a \phi \to  \K_a \K_a \phi$ is valid.
\end{lemma}
\begin{proof}
Let $Z = (\chi,\I,\rho)$. Assume we have two states $Z' = (\chi',\I',\rho')$ and 
$Z'' = (\chi'',\I'',\rho'')$ such that 
$Z \trianglelefteq_a Z'$ and $Z' \trianglelefteq_a Z''$.
We show $Z \trianglelefteq_a Z''$.

Because of $\I[a] \subseteq \I'$ and $ \I'[a] \subseteq \I''$, we obtain
\begin{equation}\label{eq:subI:1}
\I_a \subseteq \I'_a
\end{equation}
and 
\begin{equation}\label{eq:subI:2}
\I[a] \subseteq \I''.
\end{equation}

Next we show that
\begin{equation}\label{eq:chi:1}
\chi \sim^\I_a \chi''.
\end{equation}
We have $\chi(b) = \chi'(b)$ for all $b \in \I_a$ and
$\chi'(b) = \chi''(b)$ for all $b \in \I'_a$.
Using~\eqref{eq:subI:1} we get $\chi(b) = \chi''(b)$ for all $b \in \I_a$ 
and \eqref{eq:chi:1} is established.

Since $\sim_a$ is an equivalence relation on histories,  we immediately get 
\begin{equation}\label{eq:simeq:1}
\rho \sim_a \rho''.
\end{equation}

From $Z \trianglelefteq_a Z'$,  we obtain that $Z$ is $a$-consistent and $Z' \trianglelefteq_a Z''$ yields that $Z''$ is $a$-consistent.
%
This,  together with~\eqref{eq:subI:2},  \eqref{eq:chi:1},  and \eqref{eq:simeq:1} yields  $Z \trianglelefteq_a Z''$.

Thus,  $\trianglelefteq_a$ is a transitive relation and we  conclude that  $\CGS,Z \Vdash \K_a \phi$  implies $\CGS,Z'' \Vdash \phi $.
\end{proof}

We defined truthful information perspectives to satisfy~\eqref{eq:clsuf:1}.
This is a necessary condition for $Z \trianglelefteq_a Z$ to hold. 
However, it is not sufficient as the state $Z$ may not be $a$-consistent.

We have the following lemma,  which states that the Truth axiom for $\K_a$ holds at $a$-consistent states with truthful information perspectives.
This is similar to the situation in simplicial semantics~\cite{CACHIN2025114902,GoubaultKLR2024faulty},  which provides a model of epistemic logic for distributed systems.  There, the Truth axiom only holds for correct processes. Hence,  it also has the form of~\eqref{eq:axt:1}.

\begin{lemma}\label{l:refl:1}
Let $\CGS$ be a concurrent game structure and $Z = (\chi,\I,\rho)$ be a state with a truthful information perspective.
We have
\begin{equation}\label{eq:axt:1}
\CGS,Z \Vdash \lnot K_a \bot \to ( \K_a \phi \to  \phi).
\end{equation}
\end{lemma}
\begin{proof}
Assume $\CGS,Z \Vdash \lnot K_a \bot$. 
Thus, there exists a state $Z'$ with $Z \trianglelefteq_a Z'$, for otherwise we would have $\CGS,Z \Vdash K_a \bot$.
Now $Z \trianglelefteq_a Z'$ implies that $Z$ is $a$-consistent.
Since $\sim_a$ is an equivalence relation, we find $\rho \sim_a \rho$. 
Trivially, we also have $\chi \sim^{\I}_a \chi$.
By~\eqref{eq:truthful:1} we know $\I[a] \subseteq \I$. Together, this yields  $Z \trianglelefteq_a Z$. We conclude $\CGS,Z \Vdash  \K_a \phi \to  \phi$ as desired.
\end{proof}

In general,  agents cannot perform negative introspection. 

\begin{lemma}\label{l:negIntro:1}
The schema of negative introspection is not valid,  i.e.~there is a countermodel to
\[
\lnot \K_b \phi \to \K_b \lnot \K_b \phi.
\]
\end{lemma}
\begin{proof}
We consider again the setting from Section~\ref{s:firstExamples} and let $\phi$  be the formula~$\X p$. 
Further, we let $X=(\chi_a,\emptyset, v_1)$, and $Z=(\chi_a,\I, v_1)$ with $\I:=\{ba\}$.
As in Example~\ref{ex:basic:2}, we have $\CGS_1, X \Vdash \lnot \K_b \phi$.
On the other hand,   by Example~\ref{ex:basic:1}, we get  $\CGS_1, Z \Vdash  \K_b \phi$.
Because of $X \trianglelefteq_b Z$,  we conclude  $\CGS_1, X \nVdash \K_b \lnot \K_b \X p$.
\end{proof}

\section{Hanabi}

We illustrate the workings of our logic and the importance of higher-order knowledge of strategies by studying the game Hanabi.
Hanabi is a cooperative game of imperfect information.  
Because of this particular combination, Hanabi has been proposed as a  challenge for machine learning techniques in multi-agent settings~\cite{HanabiChallenge}.  Early work on Hanabi was mainly concerned with evaluating fixed strategies~\cite{FireworksDisplay}.  Later,  Hanabi has been used to study, e.g.,  zero-shot coordination,  that is, constructing agents that can coordinate with novel partners they have not seen before~\cite{OtherPlay}.
Most recently,  Perrotin~\cite{Perrotin2025} provides an epistemic logic analysis of Hanabi that is based on (common) knowledge of game states.  In contrast to our examples,  knowledge of strategies is not relevant for that work.

The goal of the game is to play cards so as to form ordered stacks,
one for each color,  beginning with a card of rank 1 and ending with a card of rank 5. 
The game's twist is the imperfect information, which arises from each player being unable to see their own cards (but they can see the hands of the other players).  Hence,  players have to rely on receiving information from the other players.

Players take turns doing one of three actions: giving a hint,  playing a card from their hand,  or discarding a card. 
The number and types of hints that are allowed are very restricted.   However, each action taken in a play is observed by all players and can, therefore, convey implicit information. This is possible, in particular, when players have pre-established conventions and tactics.
In our example, we assume that the players use the following convention: if a player thinks that several of her cards are possible to play, she plays the leftmost of them.

We use Hanabi as a prime example of the role of higher-order knowledge when reasoning about strategies.  In particular, we investigate a standard tactic of Hanabi, the so-called \emph{finesse} move,  see~\cite{HanabiChallenge,Conventions} for more details.

We consider a simplified situation,  ignoring the color of the cards and the exact form of the hints.
There are three players $a$, $b$,  and $c$.  It is player $a$'s turn,  then player $b$'s,  and then player $c$'s.
Player $a$ sees that player $b$ has a card of rank~1 on her leftmost position and player $c$ has a card of rank $2$.
Player $a$ could now
\begin{enumerate}
\item inform player $b$ of her card of rank 1.  Then player $b$ will play that card.
\item inform player $c$ of her card of rank 2 and assume that player $b$ will infer that she must have a card 1. 
So player~$b$ will play her card 1, and then player $c$ will play her card~2.
\end{enumerate}
Of course, the second option is the better choice, as with one hint, two correct cards will be played.  
This is the essence of the \emph{finesse}-move.

We model this situation with the concurrent game structure $\CGS_2$ given in Figure~\ref{fig:hanabi:1}.
\begin{figure}[ht]
\begin{center}
\begin{tikzpicture}[scale=0.45]
\node[circle, draw] at (0,0) (v1) [label = below:$p$]{$v_1$};
\node[circle, draw] at (-3,-4) (v2) [label = below:$p$]{$v_2$};
\node[circle, draw] at (3,-4) (v3) [label = below:$p$]{$v_3$};

\node[circle, draw] at (11,0) (v4) [label = below:$\lnot p$]{$v_4$};
\node[circle, draw] at (8,-4) (v5) [label = below:$\lnot p$]{$v_5$};
\node[circle, draw] at (14,-4) (v6) [label = below:$\lnot p$]{$v_6$};

\draw[-] (v1) to (v2);
\draw[-] (v1) to (v3);
\draw[-] (v4) to (v5);
\draw[-] (v4) to (v6);

\node[circle, fill = white] at (-2.5, -1.5) (l1)[label = above:$ $]{$\alpha$};
\node[circle, fill = white] at (2.5, -1.5) (l1)[label = below:$ $]{$\beta$};
\node[circle, fill = white] at (8.5, -1.5) (l1)[label = above:$ $]{$\gamma$};
\node[circle, fill = white] at (13.5, -1.5) (l1)[label = below:$ $]{$\beta$};

\end{tikzpicture}
\caption{The concurrent game structure $\CGS_2$}\label{fig:hanabi:1}
\end{center}
\end{figure}
Proposition $p$ means that player~$b$ has a card~1.
The positions $v_1$ and $v_4$ are two possible positions before $a$'s move.  At position $v_1$, player $b$ does have a card 1,  while she does not at~$v_4$.  The positions $v_2$, $v_3$,  $v_5$,  and $v_6$ represent the different possible results of $a$'s action.

Action $\alpha$ means that player $a$ gives player $b$ the information that $b$ has a card 1.
Action~$\beta$ means that player $a$ gives player $c$ the information that she has a card 2.
Action $\gamma$ means that player $a$ gives player $b$ the information that $b$ does not have a card 1.
As in our previous examples,  only the action of player $a$ matters, and we consider $\alpha$,  $\beta$, and $\gamma$ also as joint actions.

All actions are public, and each agent can distinguish them.  In particular, we have
\begin{equation}\label{eq:hanabi:1}
\alpha \not\sim_b \beta \text{ and }  \beta \not\sim_b \gamma\text{ and } \alpha \not\sim_b \gamma .
\end{equation}

Player $a$ can see player $b$'s cards.  Thus,  $a$ can distinguish the different states,  i.e.
\begin{equation}\label{eq:paKnows:1}
v_i \not\sim_a v_j \quad\text{for $i \neq j$}.
\end{equation}
Player $b$ does not know whether she has a card 1. Thus 
\begin{equation}\label{eq:paKnows:2}
v_1 \sim_b v_4.  
\end{equation}
If player~$a$ informs player $b$,  then she knows whether she has a card 1. Thus $v_2 \not\sim_b v_5$.  However, if player $a$ gives information to player $c$, then $b$ still does not know whether she has a card 1. Hence 
\begin{equation}\label{eq:paKnows:3}
v_3\sim_b v_6.
\end{equation}
Player $a$'s strategy is as follows.
\begin{enumerate}
\item If player $b$ does not have a card 1, then $a$ gives $b$ the information that~$b$ does not have a card 1,  i.e.~she performs action $\gamma$.
\item If player $b$ has a card 1, then $a$ informs $c$,  i.e.~she performs action $\beta$.
\end{enumerate}
We denote this strategy by $\sigma$. That is $\sigma(v_1):=\beta$ and $\sigma(v_4):=\gamma$.
Further, we let the assignemt $\chi$ be such that $\chi(a):=\sigma$.
We consider the history $\rho:= v_1 \beta v_3$, i.e.~player~$b$ has a card $1$ and player $a$ gives information to $c$.


Assume that player $b$ knows player $a$'s strategy.  We show that in this situation,  $b$ can infer from observing action $\beta$ that she must have a card of rank~1.

\begin{example}\label{ex:knowledge:1}
We consider an information perspective $\I$ with $ba \in \I$ and let $Z$ be the state $(\chi, \I, \rho)$.
We find that in this situation,  player $b$ knows that she has a card 1,  although player~$a$ informed player~$c$.
Formally, we have $\CGS_2, Z \Vdash \K_b p$.
Indeed, let $Z'=(\chi',\I', \rho')$ be any state with $Z \trianglelefteq_b Z'$.
This implies $\chi \sim^\I_b \chi'$, which by $a\in \I_b$ yields $\chi'(a)=\chi(a)=\sigma$.

Since $\rho'$ must be consistent with $\chi'$ under $\I'_b$, we find
\[
\rho' = v_1 \beta v_3  \quad\text{or}\quad \rho' = v_4 \gamma v_5.
\]
Since $\rho \sim_b \rho'$,  we find by~\eqref{eq:hanabi:1}
\[
\rho' = v_1 \beta v_3  \quad\text{or}\quad \rho' = v_4 \beta  v_6.
\]
Hence, we conclude $\rho' = v_1 \beta v_3$.
Therefore,  $\CGS_2, Z' \Vdash p$ and $\CGS_2, Z \Vdash \K_b p$.
\end{example}



Now, we assume that player $b$ does not know player $a$'s strategy.  In this situation, player~$b$ does not know that she has a card of rank 1 unless player $a$ tells her so.

\begin{example}\label{ex:noknowledge:1}
Formally,  we consider an information perspective $\I$ with $ba \notin \I$ and let $Z$ be the state $(\chi, \I, \rho)$.  
We will show  $\CGS_2, Z \nVdash \K_b p$.

First, we observe that $Z$ is  $b$-consistent.
Now let $\sigma'$ be a strategy such that $\sigma'(v_1) = \beta$ and $\sigma'(v_4) = \beta$. 
Further let $\chi'$ be an assignment with $\chi'(a)=\sigma'$.
Because of $ba \notin \I$  we find $a \notin \I_b$. 
Therefore $\chi \sim^\I_b \chi'$.
We set $\rho' = v_4 \beta v_6$
and observe by \eqref{eq:paKnows:2} and \eqref{eq:paKnows:3} that $\rho \sim_b \rho'$.
Further we have that~$\rho'$ is consistent with $\chi'$ under $\I_b$, which means that $ (\chi', \I', \rho')$ is $b$-consistent.
Therefore, we obtain  $(\chi, \I, \rho) \trianglelefteq_b   (\chi', \I', \rho')$.  
We have $ (\chi', \I', \rho') \nVdash p$ and thus conclude $\CGS_2, Z \nVdash \K_b p$.
\end{example}
Let us again consider Example~\ref{ex:knowledge:1} where we established that 
\begin{multline}\label{eq:precondition:1}
\text{if $b$ knows $a$'s strategy,  then  $b$ will infer from observing action $\beta$}\\
\text{ that she must have a card of rank 1.}
\end{multline}
However, note that this is not enough for player $a$ to safely choose action $\beta$.
What is required is that player $a$ knows~\eqref{eq:precondition:1}.  Only then can $a$  be sure that if she informs $c$,  player~$b$ will play her card 1, and thus informing $c$ about her card 2 is a good choice.


Hence, this is a situation where higher-order knowledge is relevant.  Namely,  we have that
if player $a$ knows that player~$b$ knows $a$'s strategy,  then player~$a$ knows that if she informs player $c$,  then player $b$ will know that she has a card~1.
%
%
%
%
We can formalize this situation as follows.

\begin{example}
We let $\I$ be an information perspective with $aba \in \I$ and consider the state $Z=(\chi, \I, v_1)$.
Let $Y= (\chi', \I', \rho')$ be an arbitrary state with $Z \trianglelefteq_a Y$.  Thus $\rho \sim_a \rho'$.  Therefore by~\eqref{eq:paKnows:1}, we get $\rho'=v_1$.
Moreover, we have \mbox{$\chi \sim^{\I}_a \chi'$}.  Because of $a \in \I_a$,  we get 
\[
\chi'(a)=\chi(a)=\sigma.
\]
Hence $\X_{\CGS_2}^{\chi'}{\rho'}=v_1 \beta v_3$.
So we let $X$ be the state $(\chi',\I',  v_1 \beta v_3)$.  We observe $ba \in \I[a] \subseteq \I'$.
Therefore,  we obtain $\CGS_2, X \Vdash \K_b p$ as in Example~\ref{ex:knowledge:1}.  Hence 
\[
\CGS_2, Y \Vdash \X \K_b p.
\]
Since $Y$ was arbitrary with  $Z \trianglelefteq_a Y$, we conclude 
\[
\CGS_2, Z \Vdash \K_a \X \K_b p.
\]
\end{example}

The reasoning of the previous example is only possible if agent $a$ knows that $b$ is informed about $a$'s strategy. 
Otherwise,  $a$ must take into account a situation where $b$ is not informed about $a$'s strategy, and thus $a$ considers the situation of Example~\ref{ex:noknowledge:1} possible.  Hence, even though $b$ will know that she has a card~1,  player~$a$ will not know that $b$ knows this.

\begin{example}
We let $\I$ be an information perspective with $ba \in \I$ but $aba \notin \I$  and consider an $a$-consistent state $Z=(\chi, \I, v_1)$.
We set $Y=(\chi, \I[a], v_1)$ and find $Z \trianglelefteq_a Y$ and $ba \notin \I[a]$.
 Further, we have  $\X_{\CGS_2}^{\chi}{v_1}=v_1 \beta v_3$ and let $X$ be the state  $(\chi,\I[a],  v_1 \beta v_3)$. 
As in Example~\ref{ex:noknowledge:1}, we find  
\[
\CGS_2, X \nVdash \K_b p.
\]
Hence $\CGS_2, Y \nVdash \X \K_b p$ and finally $\CGS_2, Z \nVdash \K_a \X \K_b p$. 
\end{example}

\section{Common knowledge}

In this section,  we investigate common knowledge of strategies. 
Agent $a$'s strategy is common knowledge among a group of agents $G$ if all agents in $G$ are informed about it, all agents in $G$ know that all agents in $G$ are informed about it, all agents in $G$ know that all agents in $G$ know that all agents in $G$ are informed about it, and so on.
Hence,  an information perspective $\I$ models that agent $a$'s strategy is common knowlege among $G$  if 
\[
\text{$wa \in \I$ for each $w\in G^{\geq 1}$.}
\]

We can recover negative introspection if all strategies are common knowledge.  This is captured by the notion of a fully informed information perspective.
Formally, we have the following definition.
\begin{definition}
We say that the information perspective $\I$ is \emph{fully informed} if $\I = \Ag^{\geq 2}$.
\end{definition}

A fully informed information perspective $\I$ satisfies the following important property for any agent $a$
\begin{equation}\label{eq:fully:1}
\I[a] = \I = \Ag^{\geq 2}.
\end{equation}

The following lemma states that Axiom B holds for fully informed information perspectives.
\begin{lemma}
Let $\CGS$ be a concurrent game structure and $Z = (\chi,\I,\rho)$ be a state with a fully informed information perspective. We have
\[
\CGS,Z \Vdash  \phi \to  \K_a  \lnot \K_a \lnot \phi .
\] 
\end{lemma}
\begin{proof}
We show that the accessibility relation in symmertric,  i.e.~for each state $Z' = (\chi',\I',\rho')$ with $Z \trianglelefteq_a Z'$, 
we have $Z' \trianglelefteq_a Z$.
Indeed,  we observe that by~\eqref{eq:fully:1}, we get $\I' = \I$.
Thus we have $\chi \sim^{\I'}_a \chi'$ and $\I'[a] \subseteq \I$,  which yields
$Z' \trianglelefteq_a Z$.

Now assume $\CGS,Z \Vdash  \phi$. 
Let $Z'$ be any state with $Z \trianglelefteq_a Z'$.  
We have to show  $\CGS,Z' \Vdash   \lnot \K_a \lnot \phi$, that is we have to find a state $Z''$ with  $Z' \trianglelefteq_a Z''$ and $\CGS,Z'' \Vdash  \phi$.  By the above observations, we can use $Z$ as  $Z''$.
\end{proof}

As usual,   this yields negative introspection.

\begin{lemma}
Let $\CGS$ be a concurrent game structure and $Z = (\chi,\I,\rho)$ be a state with a fully informed information perspective.
We have
\[
\CGS,Z \Vdash  \lnot \K_a \phi \to  \K_a  \lnot \K_a  \phi .
\] 
\end{lemma}
\begin{proof}
From the proof of Lemma~\ref{l:ax4:1}, we know that the accessibility relation is transitive, and the previous proof tells us that it is symmetric.  As usual, this yields that it is Euclidean, i.e.
\begin{equation}\label{eq:Euclid:1}
Z \trianglelefteq_a X \text{ and } Z \trianglelefteq_a Y \quad\text{implies}\quad Y \trianglelefteq_a X.
\end{equation}

Now assume  $\CGS,Z \Vdash  \lnot \K_a \phi$. Thus there is a state $X$ with $Z \trianglelefteq_a X$ and $\CGS,  X \Vdash  \lnot \phi$.  
Let $Y$ be an arbitrary state with  $Z \trianglelefteq_a Y$.  
By~\eqref{eq:Euclid:1},  we find $Y \trianglelefteq_a X$ and thus 
$\CGS, Y \Vdash  \lnot  \K_a  \phi$. 
Since $Y$ was arbitrary with  $Z \trianglelefteq_a Y$,  we conclude
\[
\CGS,Z \Vdash  \K_a \lnot  \K_a  \phi. \qedhere
\] 
\end{proof}


We illustrate the need for common knowledge of strategies with the problem of binary consensus.
In this problem,   we have two agents,  each of whom has an input value that can be either 0 or 1, and they have to agree on a common value,  which has to be one of their input values.  In this context, agreeing on a value means that the value is common knowledge for the two agents.

We consider the following scenario.
There are two agents, $a$ and $b$, and the group $G$ consists of these two agents. 
In the following, we will just say common knowledge when we mean common knowledge among $G$.
One agent has 0 as the input value; the other agent has an input value of 1. 
They exchange their values using reliable communication.  So, it is common knowledge 
that the agents have 0 and 1 as their respective input values.
We represent this by the position $v_1$ in a concurrent game structure $\CGS_3$ that has the same form as the one given in  Figure~\ref{fig:CGS1}.

Now, agent $a$ has to choose its output value, and we only consider the question of whether her choice will be common knowledge.
Action $\alpha$ means agent $a$ uses the minimum function to choose its output value.
Action $\beta$ means agent $a$ uses the maximum function.
So $p$ stands for $a$ decided for 0, and $q$ means that $a$ decided for 1 as the output value.

We assume agent $a$ uses the minimum function,  i.e.~her strategy is to perform action~$\alpha$.
We let $\chi$ be such that \mbox{$\chi(a)(v_1)=\alpha$} and  $\chi'$ be such that $\chi'(a)(v_1)=\beta$.
Agent $a$ does not announce which output value she has chosen.  That is in~$\CGS_3$,  we have $v_2 \sim_b v_3$ and $\alpha  \sim_b \beta$.
Let $\rho := v_1 \alpha v_2$ and  $\rho' := v_1 \beta v_3$.  We find that $\rho \sim_b \rho'$.

We now show that common knowledge of $a$'s strategy is necessary for common knowledge of $p$.
\begin{theorem}\label{th:ck:1}
Let $\I$ be an information perspective such that $a$'s strategy is not common knowledge.
We have
\[
\CGS_3,  (\chi, \I,\rho) \nVdash \CK p.
\]
\end{theorem}
\begin{proof}
We define sequences $w_i \in \Ag^{\geq 2}$ for $i \leq 1$ as follows:
$w_1 := ba$,  $w_2 := aba$,  $w_3 := baba$,  and so on.
For a sequence $w= x_n \cdots x_2 x_1 \in  \Ag^{\geq 2}$, we let $\phi_w$ be the formula
$\K_{x_n} \cdots \K_{x_2} p$.  Further, we write $\phi_i$ for  $\phi_{w_i}$.
For instance,  we have  $\phi_1 = \K_b p$ and $\phi_2 = \K_a \K_b p$.

We show
\[
\text{for all information perspectives $\I$},  w_i \notin \I \text{ implies }  \CGS_3,(\chi, \I,\rho) \nVdash \phi_i 
\]
by induction on $i$.

Case $i=1$.
Since $ba \notin \I$, we $a \notin \I_b$.
Therefore,  $\chi \sim^{\I}_b \chi'$. 
This yields $ (\chi, \I,\rho) \trianglelefteq_b  (\chi', \I[b],\rho')$.
Because of $\CGS_3, (\chi', \I[b],\rho') \nVdash p$, we conclude  
\[
  \CGS_3, (\chi, \I,\rho) \nVdash \K_b p.
\]
Case $i=j+1$.
Let $x\in \Ag$ be such that $w_i = x w_j$.
We have 
\[
(\chi, \I,\rho) \trianglelefteq_x  (\chi, \I[x],\rho).
\]
By the definition of $\I[x]$ and $w_i \notin \I$,  we find $w_j \notin \I[x]$.
Thus,  by I.H.,  we obtain 
\[
\CGS_3, (\chi, \I[x],\rho) \nVdash \phi_j.
\]
Therefore,  $  \CGS_3, (\chi, \I,\rho) \nVdash \K_x \phi_j$, which is  
\[
\CGS_3,(\chi, \I,\rho) \nVdash \phi_i.
\]
If $a$'s strategy is not common knowledge, there exists $i$ with $w_i \notin \I$ and thus $ \CGS_3,(\chi, \I,\rho) \nVdash \phi_i$.
Since $\CK p \to \phi_i$ is valid, we conclude $\CGS_3,(\chi, \I,\rho) \nVdash \CK p$.
\end{proof}

Now, we show that if $a$'s strategy is common knowledge,  then $p$ is common knowledge.
\begin{theorem}\label{th:ck:2}
Let $\I$ be such that $a$'s strategy is common knowledge.
We have
\[
\CGS_3,  (\chi, \I,\rho) \Vdash \CK p.
\]
\end{theorem}
\begin{proof}
By induction on $i$ we show that for all states $Z_0, \ldots, Z_i$ such that
\[
(\chi, \I,\rho)  =: Z_0 \trianglelefteq_{x_1} Z_1  \trianglelefteq_{x_2} \cdots \trianglelefteq_{x_i} Z_i ,
\]
we have 
\begin{enumerate}
\item
$\CGS_3,Z_i \Vdash p$ and
\item
for $Z_i=(\sigma^i, \I^i, \rho^i)$,  agent~$a$'s strategy is common knowledge in $\I^i$.
\end{enumerate}

Case $i=0$. 
We have $\last(\rho) = v_2$ and hence $\CGS_3,(\chi, \I,\rho)  \Vdash p$.
By assumption, we also have that $a$'s strategy is common knowledge in $\I$.

Case $i=j+1$.
We let $Z_j = (\chi^j, \I^j, \rho^j)$.
By I.H.~we find $\CGS_3,Z_j  \Vdash p$,  which means that $\last(\rho^j) = v_2$, and
\begin{equation}\label{eq:IH:1}
\text{agent $a$'s strategy is common knowledge in $\I^{j}$.}
\end{equation}
Thus,  $a \in \I^{j}_{x_i}$.
Since $\rho^j$ is consistent with $\chi^j$ under~$\I^j_{x_i}$, we find $\chi^j(a)(v_1) = \alpha$.
We have \mbox{$\chi^j \sim^{\I^j}_{x_i} \chi^i$}.
Therefore, $\chi^i(a)(v_1)= \chi^j(a)(v_1) = \alpha $.
Since $\rho^i$ is consistent with $\chi^i$ under $\I^i_{x_i}$,  we get  $\last(\rho^i) = v_2$.
Hence, we conclude 
\[
\CGS_3,Z_i \Vdash p.
\]

To show the second claim,  we let $w=w_h\cdots w_1 \in \Ag^{\geq 1}$.
If $w_h \neq x_i$,  then by~\eqref{eq:IH:1} we have $x_{i}wa \in \I^j$ and thus 
\[
wa \in \I[x_i] \subseteq \I^i.
\]
If $w_h = x_i$,  then  by~\eqref{eq:IH:1} we have $wa \in \I^j$ and again 
\[
wa \in \I[x_i] \subseteq \I^i.
\]
Hence, $a$'s strategy is common knowledge in $\I^i$.

We have shown that for any state $Z_i$ that is reachable from $(\chi, \I,\rho) $, we have $\CGS_3,Z_i \Vdash p$.  Therefore  $\CGS_3, (\chi, \I,\rho)  \Vdash \CK p$ as desired.
\end{proof}


\section{Model checking}

Finally, we investigate the model checking problem. 
We will show that it is decidable for $\lang$-formulas.

We begin with the following auxiliary notion. 
The $\K$-depth of an $\lang$-formula $\phi$,  in symbols $\kd(\phi)$,  is inductively defined by:
\begin{align*}
&\kd(p) := 0  &   & \kd(\phi \to \psi):=\max(\kd(\phi),\kd(\phi)) \\
& \kd(\bot):=0  &
&\kd(\K_a \phi)  := \kd(\phi)+1  \\
&  \kd(\X \phi) := \kd(\phi).
\end{align*}

For a set $X \subseteq \Ag^{*}$ and a natural number $n$,
we set 
\[
X \! \upharpoonright_{n}:=\{ w \in X \ |\ \len(w) \leq n\}.
\]

\begin{lemma}\label{l:finiteI:1}
Let $\CGS$ be a concurrent game structure, $Z=(\chi,\I,\rho)$ be a state,   and $\phi$ be an $\lang$-formula.
For any information persprective $\I'$ with $\I  \! \upharpoonright_{\kd(\phi)+2} = \I'  \! \upharpoonright_{\kd(\phi)+2}$,  we have
\[
\CGS,  (\chi, \I,\rho) \Vdash \phi \quad\text{if{f}}\quad \CGS,  (\chi, \I',\rho) \Vdash \phi.
\]
\end{lemma}
\begin{proof}
We show by induction on $n$ and subinduction on $\phi$: if $\I  \! \upharpoonright_{n+2} = \I'  \! \upharpoonright_{n+2}$, then for any $\lang$-formula $\phi$ with $\kd(\phi)\leq  n$,  we have 
\begin{equation}\label{eq:mc:1}
\CGS,  (\chi, \I,\rho) \Vdash \phi \quad\text{if{f}}\quad \CGS,  (\chi, \I',\rho) \Vdash \phi.
\end{equation}

Case $n=0$.   
The formula $\phi$ does not contain any $\K$-modality, and \eqref{eq:mc:1} trivially holds.

Case $n= k+1$.
For the subinduction, we only show the case $\phi = \K_a \psi$.
We show the direction from left to right by contraposition.  The other direction is analogous.
Assume 
\[
\CGS,  (\chi, \I',\rho) \nVdash \K_a \psi.
\]
Thus,  there exists an information perspective $\I'_0$ such that
\begin{equation}\label{eq:mc:6.1}
\CGS,  (\chi, \I'_0,\rho) \nVdash \psi 
\end{equation}
and
\begin{equation}\label{eq:mc:6.2}
(\chi, \I',\rho)  \trianglelefteq_a    (\chi, \I'_0,\rho).
\end{equation}
Hence
\begin{equation}\label{eq:mc:5}
\I'[a] \subseteq \I'_0.
\end{equation}
Now we set 
\begin{equation}\label{eq:mc:4}
\I_0:= \I[a] \cup \I'_0.
\end{equation}
We trivially have
\begin{equation}\label{eq:mc:2}
\I[a] \subseteq \I_0.
\end{equation}
Further we have 
\begin{equation}\label{eq:mc:3}
\I_0  \!\upharpoonright_{n+1} \ =\ \I'_0 \!\upharpoonright_{n+1}.
\end{equation}
Indeed,  $\supseteq$ holds by definition~\eqref{eq:mc:4}.  For $\subseteq$,  we assume 
\[
w \in \I_0  \!\upharpoonright_{n+1}
\]
and distinguish the cases according to~\eqref{eq:mc:4}:
If $w \in \I'_0$, then $w \in  \I'_0  \!\upharpoonright_{n+1}$ and we are done.
If $w \in \I[a]$, then  we distinguish:
\begin{enumerate}
\item 
$aw \in \I$.  We have $aw \in  \I  \!\upharpoonright_{n+2} \ = \ \I'  \!\upharpoonright_{n+2}$. 
Thus  $w \in \I'[a]$ and by~\eqref{eq:mc:5}, we get $w \in \I'_0$.  Hence $w \in  \I'_0 \!\upharpoonright_{n+1}$.
\item 
$w \in \I$ and $w=aw'$ for some $w'$.  Similar to the previous case.
\end{enumerate}
From \eqref{eq:mc:3} and \eqref{eq:mc:6.1}, we get by I.H.~that 
$\CGS,  (\chi, \I_0,\rho) \nVdash \psi$.
Further we have $\I  \! \upharpoonright_{2} = \I'  \! \upharpoonright_{2}$  by assumption and $\I_0  \! \upharpoonright_{2} = \I'_0  \! \upharpoonright_{2}$ by \eqref{eq:mc:3}.
We obtain $\I_a = \I'_a$ and $(\I_0)_a = (\I'_0)_a$, respectively.
Hence,  using \eqref{eq:mc:6.2}, we find that
$ (\chi, \I,\rho)$ and  $(\chi, \I_0,\rho)$ are $a$-consistent.
Together with \eqref{eq:mc:2},  this yields
\[
(\chi, \I,\rho)  \trianglelefteq_a    (\chi, \I_0,\rho). 
\]  
We finally conclude
$\CGS,  (\chi, \I,\rho) \nVdash \K_a \psi$.
\end{proof}

The \emph{duration} of a history is the number of joint actions occurring in it, i.e.~the duration of $v_0 \alpha_1 v_1 \ldots \alpha_n v_n$ is $n$.  The duration of a state $(\chi,\I,\rho)$ is the duration of $\rho$.

Let $Z$ be a state with duration $n$.  By Definition~\ref{def:truth:1}, we find that
\begin{enumerate}
\item
to evaluate a formula  of the form  $\X \phi$ at $Z$, we have to evaluate $\phi$ at a state with duration $n+1$;
\item 
to evaluate a formula  of the form  $\K_a \phi$ at $Z$, we have to evalute $\phi$ at states with duration $n$.
\end{enumerate}
Let $\phi$ be an $\lang$-formula with a maximal nesting depth of $\X$-operators of~$m$.
We find that to evaluate $\phi$ at $Z$, we only have to consider states with a duration of at most  $n+m$.

This,  Lemma~\ref{l:finiteI:1}, and the fact that a
concurrent game structure consists only of finitely many actions and finitely many positions
together yield that to evaluate an $\lang$-formula at a given state, we only have to consider finitely many other states. Therefore, we obtain the following corollary.
\begin{corollary}
Model checking is decidable for $\lang$-formulas.
\end{corollary}

It is an open question whether model checking for $\langCK$-formulas is also decidable.
A consequence of Theorems~\ref{th:ck:1} and \ref{th:ck:2} is that the common knowledge part of an information perspective matters for evaluating $\langCK$-formulas.  Thus, for  $\langCK$-formulas,   a restriction to a finite information perspective will not work, and Lemma~\ref{l:finiteI:1} cannot be adapted to general information perspectives.  

One approach to show decidability of model checking for $\langCK$-formulas could be to consider information perspectives~$\I$ for which there are a finite information perspective $\I_1$ and  a set $X \subseteq \Ag$ such that
\[
\I = \I_1 \cup \{ wx \ |\ w\in \Ag^{\geq 2} \text{ and } x \in X  \}.
\]
This provides a finite representation of an information perspective with common knowledge and could make it possible to establish decidable model checking for $\langCK$.

Another question is whether the proof of Lemma~\ref{l:finiteI:1} can be adapted to the case where only truthful information perspectives are allowed.  The problem is that $\I_0$ as defined in~\eqref{eq:mc:4} need not be truthful in general.

Finally, we may ask what happens if we add the until-operator to $\lang$.  However,  it is proved in \cite{MeydenShilov99} that in distributed systems with perfect recall,  model checking is undecidable for a language with operators for until and common knowledge.

\section{Conclusion}

In this work,  we introduced a model that can express first-order,  higher-order, and common knowledge of strategies. 
We studied how higher-order knowledge of strategies affects agents' strategic abilities in the game Hanabi, and we showed that common knowledge of strategies is necessary to solve the consensus task.  Finally, we established the decidability of the model checking problem for the language without common knowledge.

Two important questions left for future work are studying the model checking problem in the presence of common knowledge
and extending our approach to strategy logic,  i.e.~to a language that treats strategies as explicit first-order objects.


\begin{thebibliography}{10}

\bibitem{10.1145/585265.585270}
R.~Alur, T.~A. Henzinger, and O.~Kupferman.
\newblock Alternating-time temporal logic.
\newblock {\em Journal of the ACM}, 49(5):672–--713, 2002.

\bibitem{Artemov2014}
S.~Artemov.
\newblock On definitive solutions of strategic games.
\newblock In A.~Baltag and S.~Smets, editors, {\em Johan van Benthem on Logic
  and Information Dynamics}, pages 487--507. Springer International Publishing,
  2014.

\bibitem{Artemov2024}
S.~Artemov.
\newblock Beyond knowledge of the model.
\newblock In Y.~Weiss and R.~Birman, editors, {\em Saul Kripke on Modal Logic},
  pages 23--41. Springer International Publishing, 2024.

\bibitem{HanabiChallenge}
N.~Bard, J.~N. Foerster, S.~Chandar, N.~Burch, M.~Lanctot, H.~F. Song,
  E.~Parisotto, V.~Dumoulin, S.~Moitra, E.~Hughes, I.~Dunning, S.~Mourad,
  H.~Larochelle, M.~G. Bellemare, and M.~Bowling.
\newblock The {H}anabi challenge: A new frontier for {AI} research.
\newblock {\em Artificial Intelligence}, 280:103216, 2020.

\bibitem{ijcai2021p246}
F.~Belardinelli, S.~Knight, A.~Lomuscio, B.~Maubert, A.~Murano, and S.~Rubin.
\newblock Reasoning about agents that may know other agents’ strategies.
\newblock In Z.-H. Zhou, editor, {\em Proceedings of the Thirtieth
  International Joint Conference on Artificial Intelligence, {IJCAI-21}}, pages
  1787--1793. International Joint Conferences on Artificial Intelligence
  Organization, 2021.

\bibitem{CachinBook}
C.~Cachin, R.~Guerraoui, and L.~Rodrigues.
\newblock {\em Introduction to Reliable and Secure Distributed Programming}.
\newblock Springer, 2011.

\bibitem{CACHIN2025114902}
C.~Cachin, D.~Lehnherr, and T.~Studer.
\newblock Synergistic knowledge.
\newblock {\em Theoretical Computer Science}, 1023:114902, 2025.

\bibitem{CHATTERJEE2010677}
K.~Chatterjee, T.~A. Henzinger, and N.~Piterman.
\newblock Strategy logic.
\newblock {\em Information and Computation}, 208(6):677--693, 2010.

\bibitem{FireworksDisplay}
C.~Cox, J.~De~Silva, P.~Deorsey, F.~H.~J. Kenter, T.~Retter, and J.~Tobin.
\newblock How to make the perfect fireworks display: Two strategies for
  {H}anabi.
\newblock {\em Mathematics Magazine}, 88(5):323--336, 2015.

\bibitem{fhmv95}
R.~Fagin, J.~Y. Halpern, Y.~Moses, and M.~Y. Vardi.
\newblock {\em Reasoning about Knowledge}.
\newblock MIT Press, 1995.

\bibitem{GoubaultKLR2024faulty}
{\'E}.~Goubault, R.~Kniazev, J.~Ledent, and S.~Rajsbaum.
\newblock Simplicial models for the epistemic logic of faulty agents.
\newblock {\em Bolet{\'i}n de la Sociedad Matem{\'a}tica Mexicana}, 30(3),
  2024.

\bibitem{OtherPlay}
H.~Hu, A.~Lerer, A.~Peysakhovich, and J.~Foerster.
\newblock "{O}ther-play" for zero-shot coordination.
\newblock In {\em Proceedings of the 37th International Conference on Machine
  Learning}. JMLR.org, 2020.

\bibitem{DBLP:conf/kr/MaubertM18}
B.~Maubert and A.~Murano.
\newblock Reasoning about knowledge and strategies under hierarchical
  information.
\newblock In M.~Thielscher, F.~Toni, and F.~Wolter, editors, {\em Principles of
  Knowledge Representation and Reasoning 2018}, pages 530--540. {AAAI} Press,
  2018.

\bibitem{10.1145/2631917}
F.~Mogavero, A.~Murano, G.~Perelli, and M.~Y. Vardi.
\newblock Reasoning about strategies: On the model-checking problem.
\newblock {\em ACM Trans. Comput. Logic}, 15(4):1--47, 2014.

\bibitem{10.1093/logcom/12.1.149}
M.~Pauly.
\newblock A modal logic for coalitional power in games.
\newblock {\em Journal of Logic and Computation}, 12(1):149--166, 2002.

\bibitem{Perrotin2025}
E.~Perrotin.
\newblock A logical analysis of {Hanabi}.
\newblock In {\em Proceedings of the 39th AAAI Conference on Artificial
  Intelligence}, 2025.

\bibitem{Saffidine_Schwarzentruber_Zanuttini_2018}
A.~Saffidine, F.~Schwarzentruber, and B.~Zanuttini.
\newblock Knowledge-based policies for qualitative decentralized {POMDP}s.
\newblock {\em Proceedings of the AAAI Conference on Artificial Intelligence},
  32(1), 2018.

\bibitem{HoekW03a}
W.~van~der Hoek and M.~J. Wooldridge.
\newblock Cooperation, knowledge, and time: Alternating-time temporal epistemic
  logic and its applications.
\newblock {\em Studia Logica}, 75(1):125--157, 2003.

\bibitem{MeydenShilov99}
R.~van~der Meyden and N.~V. Shilov.
\newblock Model checking knowledge and time in systems with perfect recall.
\newblock In C.~P. Rangan, V.~Raman, and R.~Ramanujam, editors, {\em
  Foundations of Software Technology and Theoretical Computer Science}, pages
  432--445, 1999.

\bibitem{Conventions}
Z.~Volrath.
\newblock H-group conventions, 2024.
\newblock \url{https://hanabi.github.io}, last accessed 22 April 2025.

\end{thebibliography}

\end{document}